\documentclass[11pt,letter]{article}

\usepackage[english]{babel}
\usepackage[utf8x]{inputenc}
\usepackage{amsmath,amssymb,amsfonts,amsthm}
\usepackage[margin=1in]{geometry}
\usepackage{geometry}
\usepackage{parskip}
\newtheorem{theorem}{Theorem}
\newtheorem{proposition}{Proposition}
\newtheorem{lemma}{Lemma}
\theoremstyle{remark}
\usepackage[hang]{footmisc}
\newtheorem*{remark}{Remark}
\usepackage{enumitem}
\usepackage[dvipsnames,svgnames,x11names,hyperref]{xcolor}
\usepackage[longnamesfirst]{natbib}
\usepackage{graphicx}
\usepackage[hyphens,spaces,obeyspaces]{url}
\usepackage{hyperref}
\hypersetup{
    colorlinks=true,
    linkcolor=MidnightBlue,
    filecolor=MidnightBlue,      
    urlcolor=BrickRed,
    citecolor=MidnightBlue
}
\usepackage{cleveref}

\title{Concavity and Convexity of Order Statistics in Sample Size}
\author{Mitchell Watt\footnote{Department of Economics, Stanford University. Email: \href{mailto:mwatt@stanford.edu}{mwatt@stanford.edu}. I gratefully acknowledge the support of the Koret Fellowship, the Ric Weiland Graduate Fellowship in the Humanities and Sciences, and the Gale and Steven Kohlhagen Fellowship in Economics. }}

\begin{document}
\maketitle

\begin{abstract}
We show that the expectation of the $k^{\mathrm{th}}$-order statistic of an i.i.d. sample of size $n$ from a monotone reverse hazard rate (MRHR) distribution is convex in $n$ and that the expectation of the $(n-k+1)^{\mathrm{th}}$-order statistic from a monotone hazard rate (MHR) distribution is concave in $n$ for $n\ge k$. We apply this result to the analysis of independent private value auctions in which the auctioneer faces a convex cost of attracting bidders. In this setting, MHR valuation distributions lead to concavity of the auctioneer's objective. We extend this analysis to auctions with reserve values, in which concavity is assured for sufficiently small reserves or for a reasonably large number of bidders.
\end{abstract}

\section{Introduction}
In this paper, I provide conditions under which the order statistics of an i.i.d. sample of size $n$ are convex or concave in $n$.

\section{Mathematical Preliminaries}

Let $X_1, X_2,...,X_n$ be a sample of size $n$ drawn identically and independently from a continuous distribution with cumulative distribution function $F$ with support in $\mathbb{R}$ and let $f$ be its associated probability density function. Let $X_{1:n}\le X_{2:n} \le ... \le X_{n-1:n} \le X_{n:n}$ be the order statistics (obtained by sorting the sample from smallest to largest) and denote by $\mu_{k:n}$ the expected $k^{\mathrm{th}}$ order statistic of a sample of size $n$. Throughout, we assume that $F$ is such that the expectations $\mu_{k:n}$ exist (for finite $1\le k\le n$). We use the notation $\overline{F}(x)=1-F(x)$ for the reliability function associated with $F$.

The \emph{hazard rate} (or failure rate) of distribution $F$ is defined by 
\[ h(x)=\frac{f(x)}{\overline{F}(x)}.\]
The \emph{reverse hazard rate} of distribution $F$ is defined by
\[ h'(x)=\frac{f(x)}{F(x)}.\]
A distribution has \emph{monotone hazard rate (MHR)} (or increasing failure rate) if the function $h(x)$ is nondecreasing in $x$. A distribution has \emph{monotone reverse hazard rate (MRHR)} (decreasing reverse failure rate) if $h'(x)$ is nonincreasing in $x$. A distribution is \emph{log-concave} if $\log(f(x))$ is a concave function of $x$.

It is well-known (for example, by \cite{gupta2012log}) that log-concave distributions have MHR and MRHR (but the reverse implication does not hold). Log-concave distributions include the normal distribution, the uniform distributions over a convex set, the extreme value distribution, the gamma distribution (for certain shape parameters), the Weibull distribution (for certain shape parameters) and the exponential distribution (which is the unique family of distributions with a \textit{constant} hazard rate).

\section{Main Result}
We consider the expected $k^{\mathrm{th}}$ order statistic and the expected $(n-k+1)^{\mathrm{th}}$ order statistic for fixed $k$ as a function of the sample size $n\ge k$. Informally, we refer to the former as the $k^{\mathrm{th}}$ ``bottom'' order statistic and the latter as the $k^{\mathrm{th}}$ ``top'' order statistic.  It is well-known that for \emph{any} distribution function $F$, the expected minimum observation of the sample, $\mu_{1:n}$, is a nonincreasing, convex function of the sample size and the expected maximum observation, $\mu_{n:n}$, is a nondecreasing, concave function of the sample size (these results also follows from our results below). In \Cref{thm:main}, we extend this characterization to the top and bottom order statistics under assumptions on the distribution's hazard rate.

\begin{theorem} \label{thm:main}
Let $F$ be a continuous distribution and consider any fixed $k\ge 1$:
\begin{enumerate}[label=(\alph*)]
    \item If $F$ has MRHR, the expected $k^\mathrm{th}$ order statistic, $\mu_{k:n}$, is a nonincreasing, convex function of the sample size $n$ for $n\ge k$. 
    \item If $F$ has MHR, the expected $(n-k+1)^\mathrm{th}$ order statistic, $\mu_{n-k+1:n},$ is a nondecreasing, concave function of the sample size $n$ for $n\ge k$.
\end{enumerate}
In particular, if $f$ is log-concave, the conclusions of both (a) and (b) above hold.
\end{theorem}
\begin{proof}[Proof of (a)]
The bottom order statistics satisfy the identity due to \cite{david1997augmented}
\begin{equation} \Delta_{k:n}:= \mu_{k:n}-\mu_{k:n+1} = \binom{n}{k-1}\int_{-\infty}^\infty F^k(x)\overline{F}^{n-k+1}(x) \, \mathrm{d}x \label{id:david} \end{equation}
for $n\ge k$. Because the integrand on the right is nonnegative for all $x$, we see immediately that $\mu_{k:n}\ge \mu_{k:n+1}$, so that $\mu_{k:n}$ is nonincreasing in $n$. (Note that this conclusion does not depend on $F$ being MRHR.)

To show convexity, we will show that $\Delta_{k:n+1}\le \Delta_{k:n}$. By subtracting \cref{id:david} for $n$ and $n+1$, we obtain
\begin{align}
    \Delta_{k:n}-\Delta_{k:n+1} &= \binom{n}{k-1} \int_{-\infty}^\infty F^k(x)\overline{F}^{n-k+1}(x) \left[ 1-\frac{n+1}{n-k+2}\overline{F}(x) \right] \, \mathrm{d}x. \label{eq:difference}
\end{align}
Define $x^\ast$ by $F(x^\ast)=\frac{k-1}{n+1}$. Note that the integrand in \cref{eq:difference} is negative for $x<x^\ast$, zero at $x=x^\ast$ and positive for $x>x^\ast$. 

To sign the integral in \cref{eq:difference}, we use an approach similar to \cite{li2005note}, which exploits a lemma due to \cite{barlow1981statistical}, stated below for reference.\footnote{Because \citeauthor{barlow1981statistical}'s proof of \Cref{lem:measure} is not easily accessible, I provide a simple proof of (a) here assuming differentiability of $g$ (part (b) is proved similarly). Via integration by parts, we have that
\[\int_a^b g(x) \, \mathrm{d}W(x)=\int_a^b \left[g(x) \left(\frac{-\mathrm{d}}{\mathrm{d}x}\int_x^b \, \mathrm{d}W(x) \right) \right]\, \mathrm{d}x= g(a)\int_a^b \, \mathrm{d}W(x)+\int_a^b \left[g'(x) \int_x^b  \, \mathrm{d}W(x) \right] \, \mathrm{d}x, \] which is nonnegative by the assumptions on $g$ and $W$.}
\begin{lemma} \label{lem:measure}
Let $W$ be a measure on interval $(a,b)$ and $g$ a nonnegative function defined on the same interval.
\begin{enumerate}[label=(\alph*)]
    \item Suppose that $g$ is nondecreasing and $\int_t^b \, \mathrm{d}W(x) \ge 0$ for all $t\in(a,b)$. Then $\int_a^b g(x) \, \mathrm{d}W(x)\ge 0$.
    \item Suppose that $g$ is nonincreasing and  $\int_a^t \, \mathrm{d}W(x) \le 0$ for all $t\in (a,b)$. Then $\int_a^b g(x)\, \mathrm{d}W(x)\le 0$.
\end{enumerate}
\end{lemma}
Consider the integral
\[ J(t):= \int_{t}^\infty F^{k-1}(x)\overline{F}^{n-k+1}(x) \left[1-\frac{n+1}{n-k+2}\overline{F}(x)\right] f(x)\, \mathrm{d}x. \]
Note that the integrand of $J$ differs from that of \cref{eq:difference} by a multiplicative factor of the reverse hazard rate, $h'(x)$. Note that the integrand of $J(t)$ is negative for $x<x^\ast$, zero at $x=x^\ast$ and positive for $x>x^\ast$. Thus, $J(t)$ must be nonnegative for $t\ge x^\ast$, and for $t<x^\ast$, we have that $J(t)\ge J(-\infty)$. So to show $J(t)\ge 0$ for all $t$, it suffices to show that $J(-\infty)\ge 0.$

We make the change of variable $u=F(x)$ and rewrite the integral $J(-\infty)$ as
\begin{equation} \int_{0}^1 u^{k-1} (1-u)^{n-k+1}\left[1-\frac{n+1}{n-k+2}(1-u)\right] \, \mathrm{d}u. \label{eq:simplified}  \end{equation}
Applying the definition of the Beta function,
\[ \int_0^1 u^{m_1-1}(1-u)^{m_2-1} \, \mathrm{d} x = B(m_1,m_2)= \frac{(m_1-1)!(m_2-1)!}{(m_1+m_2-1)!}, \]
for positive integers $m_1$ and $m_2$,\footnote{See, for example, \cite{jeffrey2008handbook}, Section 11.1.7.} we have that
\[ J(-\infty)= \frac{(k-1)!(n-k+1)!}{(n+2)!}. \]
This is clearly nonnegative so that $J(t)\ge 0$ for all $t$. So by \Cref{lem:measure}(a), multiplying the integrand in $J(t)$ through by the inverse reverse hazard rate $1/h'(x)$ (which is nondecreasing by assumption that $F$ is MRHR), we obtain that the integral in \cref{eq:difference} is nonnegative as well. This implies $\Delta_{k:n}\ge \Delta_{k:n+1}$, so that $\mu_{k:n}$ is convex in $n$.
\end{proof}
\begin{remark}
Note that for $k=1$, the integral in \cref{eq:difference} is $ \int_{-\infty}^\infty F^{k+1}(x)\overline{F}^{n-k+1}(x) \mathrm{d}x, $ which is always nonnegative. This shows that the minimal order statistic is always convex in $n$, without any assumptions on the distribution $F$.
\end{remark}

\begin{proof}[Proof of (b)]
We apply another identity due to \cite{david1997augmented},
\[ \mu_{r:n}-\mu_{r-1:n-1}=\binom{n-1}{r-1}\int_{-\infty}^{\infty} F^{r-1}(x)\overline{F}^{n-r+1}(x) \mathrm{~d} x. \]
Letting $r=n-k+1$, we obtain
\begin{equation} \delta_{k:n}:=\mu_{n-k+1:n}-\mu_{n-k: n-1}=\binom{n-1}{n-k}\int_{-\infty}^{\infty} F^{n-k}(x)\overline{F}^{k}(x) \mathrm{~d} x. \label{id:top} \end{equation}
The integrand is nonnegative, so that $\mu_{n-k+1:n}\ge \mu_{n-k:n-1}$, that is, the $k^{\mathrm{th}}$ top order statistic is nondecreasing as a function of $n$. (Note that this conclusion does not depend on $F$ being MHR.)

To show concavity, we will show that $\delta_{k:n}\ge \delta_{k:n+1}$. By subtracting \cref{id:top} for $n$ and $n+1$, we obtain
\begin{equation} \delta_{k:n+1}-\delta_{k:n}=\binom{n-1}{n-k}\int_{-\infty}^\infty F^{n-k}(x)\overline{F}^k(x) \left[ \frac{n}{n-k+1}F(x) -1\right] \, \mathrm{d}x. \label{eq:difftop} \end{equation}
Define $x^\dagger$ by $F(x^\dagger)=\frac{n-k+1}{n}$. The integrand in \cref{eq:difftop} is negative for $x < x^\dagger$, zero at $x=x^\dagger$ and positive for $x>x^\dagger$. We will show that the integral in \cref{eq:difftop} is nonpositive.

To do so, consider the integral
\[ K(t)=\int_{-\infty}^t F^{n-k}(x)\overline{F}^{k-1}(x)\left[\frac{n}{n-k+1}F(x)-1\right] f(x) \, \mathrm{d}x, \]
which differs from that of \cref{eq:difftop} by a multiplicative factor of $h(x)$. The integrand of $K(t)$ is negative for $x<x^\dagger$, zero at $x=x^\dagger$ and positive for $x>x^\dagger$. Thus $K(t)$ is nonpositive if $t\le x^\dagger$ and $K(t)\le K(\infty)$ for $t>x^\dagger$. Thus, if we show that $K(\infty)\le 0$, we will have that $K(t)\le 0$ for all $t$.

We make the change of variable $u=F(x)$, to rewrite $K(\infty)$ as
\[ \int_0^1 u^{n-k}(1-u)^{k-1}\left[\frac{n}{n-k+1}u-1 \right] \, \mathrm{d}u, \]
which can be evaluated as
\[ K(\infty)= -\frac{ (n-k)!(k-1)! }{ (n+2)! } .\]

Since the above integral is nonpositive, $K(t)\le 0$ for all $t$. So by \Cref{lem:measure}(b), since $1/h(x)$ is nonincreasing for MHR distributions, we have that the integral in \cref{eq:difftop} is nonpositive as well. Thus, $\delta_{k:n+1}\le \delta_{k:n}$ so that $\mu_{n-k+1:n}$ is concave in $n$.

\end{proof}

\begin{remark}
Note that for $k=1$, the integral in \cref{eq:difftop} is $-\int_{-\infty}^\infty F^{n-k}(x)\overline{F}^{k+1}(x) \, \mathrm{d}x$, which is always nonpositive. This shows that the maximal order statistic is always concave in $n$, without assumptions on the distribution $F$.
\end{remark}

We end this section by noting that in the absence of assumptions on the distribution $F$ that the conclusions of \Cref{thm:main} may not hold. In particular, we consider samples from the $\mathrm{Pareto}(a,v)$ distribution for $a,v>0$, which has probability density function $f(x)=va^vx^{-v-1}\mathbb{1}_{[x\ge a]}$ and cumulative distribution function $F(x)=1-(a/x)^v$. This distribution does not have MHR as its hazard rate is a decreasing function. The expected order statistics are calculated in \cite{malik1966exact} as
\[ \mu_{k:n}=\frac{n!\Gamma(n-k-1/v+1) }{(n-k)!\Gamma(n-1/v+1)}, \]
where $\Gamma (z)=\int_0^\infty x^{z-1}e^{-x}\, \mathrm{d}x.$ For the specific example of $a=1$ and $v=3/4$, the expected second-top order statistics, $\mu_{n-1:n}$, are 3, 5.4, 8.1 for $n=2,3,4$ respectively, which violates concavity (in fact, this order statistic is a convex function of $n$ generally). By considering the transformation $X\mapsto -X$ of the same Pareto random variable, we also obtain non-convex $2^{\mathrm{nd}}$ order statistics.

\section{Auction Application}
One well-known application of the theory of order statistics is to auction theory. In the independent private values (IPV) model of a standard auction, each potential buyer $i$'s valuation for an item, $v_{i}$, is drawn identically and independently from a common knowledge distribution $F$. The realized valuation is private knowledge to buyer $i$. By the payoff and revenue equivalence theorem (see, for example, \cite{milgrom2004putting}), in any auction with $n$ bidders where the bidder with the highest bid wins (in the Bayes-Nash equilibrium, this will be the buyer for whom $v_{i}=v_{n:n}$), the expected payoff to the auctioneer is $\mu_{n-1:n}$.

Suppose that the auctioneer faces a cost associated with attracting bidders to the auction. In particular, suppose that attracting $n$ bids to the auction results in a cost $c(n)$ for the auctioneer, where $c$ is a positive, convex function of $n$.\footnote{Convexity is a common assumption of cost functions, reflecting (in this case) that it is increasingly difficult to attract marginal participants to the auction.} The cost might capture outlays associated with advertising the auction, organizing the auction or screening potential bidders. \Cref{thm:main} allows us to conclude the following about the number of bidders in the resulting auction.

\begin{proposition} \label{prop:auction}
Consider an IPV auction with a convex cost of attracting bidders. If the value distribution $F$ has MHR and compact support, the auctioneer's objective function is concave and has a finite maximizer $n^\ast$.
\end{proposition}
\begin{proof}
The auctioneer's payoff as a function of $n$ is
\[ g(n)=\begin{cases} \mu_{n-1:n}-c(n) & \text{ if }n\ge 2
\\ 0 &\text{ otherwise } \end{cases}. \]
Since $\mu_{n-1:n}$ is concave by \Cref{thm:main}, we have that $g(n)$ is concave. Since $\mu_{n-1:n}$ is bounded above by assumption, and $\lim_{n\rightarrow \infty} c(n)=\infty$, we have that $\lim_{n\rightarrow \infty} g(n)=-\infty$. Since $g(0)=0$, $g(n)$ must have a finite maximizer $n^\ast$.
\end{proof}
\begin{remark}
MHR of the value distribution is a common assumption in auction theory as it implies Myerson's ``regularity'' (see \cite{myerson1981optimal}) which results in the optimal mechanism for the seller taking the (simple) form of a second-price auction with reserve (or a posted-price mechanism if $n=1$).
\end{remark}

Many auctions have a reserve price, in which \Cref{prop:auction} does not apply. However, the result in \Cref{prop:auction} still holds in the presence of reserve prices, as shown in \Cref{prop:reserve}.

 \begin{proposition} \label{prop:reserve}
Consider the auction environment of \Cref{prop:auction}, but suppose that the auctioneer imposes a reserve price $r$ which is in the support of $F$. Then the auctioneer's objective function is increasing and concave in $n$ for sufficiently large $n$ or sufficiently small reserves (namely $n,r$ such that $F(r)\le 1-2/n$) and thus has a finite maximizer $n^\ast$.
\end{proposition}
\begin{proof}
With a reserve, the auctioneer's revenue is given by
\[ \mathrm{revenue} = \begin{cases}
X_{n-1:n} & \text{ if }r\le X_{n-1:n},
\\ r & \text{ if }X_{n-1:n}\le r \le X_{n:n},
\\ 0 & \text{ otherwise}.
 \end{cases} \]
Write $\tilde{\mu}_{n-1:n}$ for the expectation of this random variable.

We first define a `conditional' expected order statistic $\mu^-_{r:m}$ for $r \le m\le n$ as the expected value of the $r^{\mathrm{th}}$ of the (sorted) draws of $(X_i)_{i=1,...,n}$ conditioning on $m$ of the draws being no larger than the reserve, $r$. A simple calculation gives
\begin{align*} \mu^-_{n:n} &=\int_0^r 1-\frac{F^n(y)}{F^n(r)} \, \mathrm{d}y,
\\ \mu^-_{n-1:n} &= \int_0^r 1-n\frac{F^{n-1}(y)}{F^{n-1}(r)} +(n-1)\frac{F^n(y)}{F^n(r)} \, \mathrm{d}y.
\end{align*}

 By the law of iterated expectations, we have that
 \[ \tilde{\mu}_{n-1:n} = \mu_{n-1:n} + (r-\mu^-_{n-1:n-1}) nF^{n-1}(r)\overline{F}(r) - \mu^-_ {n-1:n}F^n(r).\]
Focusing on the last two terms, substituting our expressions for $\mu^-_{n:n}$ and $\mu^-_{n-1:n}$ gives
\begin{alignat*}{2}
    \tilde{\mu}_{n-1:n} &= \mu_{n-1:n} &&+ nF^{n-1}(r)\overline{F}(r)\left(r-\int_0^r 1-\frac{F^{n-1}(y)}{F^{n-1}(r)} \, \mathrm{d}y\right) 
    \\ & &&- F^n(r) \int^r_0 1-n\frac{F^{n-1}(y)}{F^{n-1}(r)} +(n-1)\frac{F^n(y)}{F^n(r)} \, \mathrm{d}y,
    \\ &= \mu_{n-1:n} &&-rF^n(r) + \int_0^r n F^{n-1}(y) -(n-1)F^n(y)\, \mathrm{d}y.
\end{alignat*}
Note that the terms in the integrand cancel with the associated terms in the integrand of $\mu_{n-1:n}$ on the domain $(0,r)$, so that we obtain the following simple expression for $\tilde{\mu}_{n-1:n}$.
\[ \tilde{\mu}_{n-1:n}= r(1-F^n(r))+\int_r^\infty 1-nF^{n-1}(y)+(n-1)F^{n}(y) \, \mathrm{d}y. \]
The first summand is clearly concave in $n$, so it suffices to focus on the integral, which we denote $I_n$.\footnote{In doing so, we are giving up a negative term and may therefore require stronger conditions on $r,n$ than is necessary, but the current proof does not easily facilitate including the term in the expression for $I_n$.} We wish to calculate the sign of $I_{n+1}-2I_n+I_{n-1}$, which if negative, will imply concavity of $\tilde{\mu}_{n-1:n}$ in $n$. A direct calculation gives
\[ I_{n+1}-2I_n+I_{n-1} = \int_r^\infty \overline{F}^2(y) F^{n-2}(y)(nF(y)-n+1) \, \mathrm{d}y. \]
We will use a similar trick to the proof of \Cref{thm:main} and consider instead the integral
\[J= \int_r^\infty  \overline{F}(y) F^{n-2}(y)(nF(y)-n+1)  \, \mathrm{d}F(y),\]
which differs by the nonincreasing factor $\overline{F}/f$. We let $u=F(y)$, and rewrite
\[J= \int_{F(r)}^1 (1-u)u^{n-2}(nu-n+1) \, \mathrm{d}u.\]
Note that the integrand of $J$ is negative for $u\le 1-1/n$, zero at $1-1/n$ and positive thereafter. Thus, we know the integral as a function of $r$  is negative, then positive for $r$ such that $F(r)$ is nearly 1 and then zero at $F(r)=1$. So in order to sign this integral, we will need to ensure that $F(r)$ is not too large as a function of $n$. This will give us the condition on the reserve in the statement of \Cref{prop:reserve}.

Direct calculation gives
\[ J= \frac{1}{n(n+1)}\left[  -1 + F^{n-1}(r) ( n^2 \overline{F}^2(r)+n \overline{F}(r)-\overline{F}(r)+1)  \right]. \]

If we substitute $F(r)=1-\frac{2}{n}$ into our expression for $J$ and simplify, we obtain 
\[ J =\frac{1}{n(n+1)}\left(-1 + (n-2)^{n-1} n^{-n} (7n-2) \right).\]
Straightforward algebra verifies that this expression is negative at $n=3$, increasing in $n$ and limits to $0$, so that $J\le 0$ for all $n$. Thus, as long as $F(r)\le 1-\frac{2}{n}$, we have that $J$ is negative. Finally, note that the signs of the integrand (negative, zero and then positive) imply that for $r\le t \le 1,$
\[ \int_r^t -nF^n(y)+(2n-1)F^{n-1}(y)+(1-n)F^{n-2}(y) \, \mathrm{d}F(y)\le J < 0, \]
under the same conditions as those above. This implies by \Cref{lem:measure}(b) that $I_{n+1}-2I_n+I_{n-1}$ is negative for any MHR $F$, where $F(r)\le 1-\frac{2}{n}$.

\end{proof}

Concavity of the auctioneer's objective may be desirable for modeling auctions as it allows an auctioneer to safely attract bidders sequentially. That is, the net benefit of attracting some bidder never depends on the ability to attract further bidders, or equivalently, the auctioneer never regrets attracting a bidder because they were not able to attract further bidders.

\bibliography{cites.bib}

\begin{thebibliography}{8}
\newcommand{\enquote}[1]{``#1''}
\expandafter\ifx\csname natexlab\endcsname\relax\def\natexlab#1{#1}\fi

\bibitem[\protect\citeauthoryear{Barlow and Proschan}{Barlow and
  Proschan}{1975}]{barlow1981statistical}
\textsc{Barlow, R. and F.~Proschan} (1975): \emph{Statistical Theory of
  Reliability and Life Testing: Probability Models}, New York: Holt, Rinehart
  and Winston.

\bibitem[\protect\citeauthoryear{David}{David}{1997}]{david1997augmented}
\textsc{David, H.} (1997): \enquote{Augmented order statistics and the biasing
  effect of outliers,} \emph{Statistics \& Probability Letters}, 36, 199--204.

\bibitem[\protect\citeauthoryear{Gupta and Balakrishnan}{Gupta and
  Balakrishnan}{2012}]{gupta2012log}
\textsc{Gupta, R.~C. and N.~Balakrishnan} (2012): \enquote{Log-concavity and
  monotonicity of hazard and reversed hazard functions of univariate and
  multivariate skew-normal distributions,} \emph{Metrika}, 75, 181--191.

\bibitem[\protect\citeauthoryear{Jeffrey and Dai}{Jeffrey and
  Dai}{2008}]{jeffrey2008handbook}
\textsc{Jeffrey, A. and H.~H. Dai} (2008): \emph{Handbook of mathematical
  formulas and integrals}, Elsevier.

\bibitem[\protect\citeauthoryear{Li}{Li}{2005}]{li2005note}
\textsc{Li, X.} (2005): \enquote{A note on expected rent in auction theory,}
  \emph{Operations Research Letters}, 33, 531--534.

\bibitem[\protect\citeauthoryear{Malik}{Malik}{1966}]{malik1966exact}
\textsc{Malik, H.~J.} (1966): \enquote{Exact moments of order statistics from
  the {P}areto distribution,} \emph{Scandinavian Actuarial Journal}, 1966,
  144--157.

\bibitem[\protect\citeauthoryear{Milgrom}{Milgrom}{2004}]{milgrom2004putting}
\textsc{Milgrom, P.} (2004): \emph{Putting Auction Theory to Work}, Cambridge
  University Press.

\bibitem[\protect\citeauthoryear{Myerson}{Myerson}{1981}]{myerson1981optimal}
\textsc{Myerson, R.~B.} (1981): \enquote{Optimal auction design,}
  \emph{Mathematics of Operations Research}, 6, 58--73.

\end{thebibliography}
\bibliographystyle{ecta}

\end{document}